\documentclass[10pt,journal,compsoc]{IEEEtran}

\ifCLASSOPTIONcompsoc
  \usepackage[nocompress]{cite}
\else
  \usepackage{cite}

\fi

\ifCLASSINFOpdf
  \usepackage[pdftex]{graphicx}
\else
\fi

\usepackage{amsmath}
\usepackage{amsthm}
\usepackage{bm}

\usepackage{algpseudocode}

\usepackage{amssymb}
\usepackage{multirow}
\begin{document}
\title{Mitigating Cascading Effects in Large Adversarial Graph Environments}

\author{James D. Cunningham,~\IEEEmembership{}
        Conrad S. Tucker~\IEEEmembership{}
\IEEEcompsocitemizethanks{\IEEEcompsocthanksitem J. Cunningham is with the Department of Mechanical Engineering, Carnegie Mellon University, Pittsburgh, PA, 15213.\protect\\
E-mail: jamescun@andrew.cmu.edu
\IEEEcompsocthanksitem C. Tucker is with the Departments of Mechanical Engineering, Machine Learning, Robotics, and Biomedical Engineering, Pittsburgh, PA, 15213.}%
\thanks{}}

\markboth{}%
{Shell \MakeLowercase{\textit{et al.}}: Counterfactual Data Augmentation for Efficient Learning in Competitive Graph-based Environments}

\IEEEtitleabstractindextext{%
\begin{abstract}
A significant amount of society's infrastructure can be modeled using graph structures, from electric and communication grids, to traffic networks, to social networks. Each of these domains are also susceptible to the cascading spread of negative impacts, whether this be overloaded devices in the power grid or the reach of a social media post containing misinformation. The potential harm of a cascade is compounded when considering a malicious attack by an adversary that is intended to maximize the cascading impact. However, by exploiting knowledge of the cascading dynamics, targets with the largest cascading impact can be preemptively prioritized for defense, and the damage an adversary can inflict can be mitigated. While game theory provides tools for finding an optimal preemptive defense strategy, existing methods struggle to scale to the context of large graph environments because of the combinatorial explosion of possible actions that occurs when the attacker and defender can each choose multiple targets in the graph simultaneously. The proposed method enables a data-driven deep learning approach that uses multi-node representation learning and counterfactual data augmentation to generalize to the full combinatorial action space by training on a variety of small restricted subsets of the action space. We demonstrate through experiments that the proposed method is capable of identifying defense strategies that are less exploitable than SOTA methods for large graphs, while still being able to produce strategies near the Nash equilibrium for small-scale scenarios for which it can be computed. Moreover, the proposed method demonstrates superior prediction accuracy on a validation set of unseen cascades compared to other deep learning approaches.
\end{abstract}

}

\maketitle

\IEEEdisplaynontitleabstractindextext

\IEEEpeerreviewmaketitle

\ifCLASSOPTIONcompsoc
\IEEEraisesectionheading{\section{Introduction}\label{sec:introduction}}
\else
\section{Introduction}
\label{sec:introduction}
\fi

\IEEEPARstart{C}{omplex} networks are often susceptible to large cascading impacts resulting from the failure of only a small number of their elements. If a small set of users share a social media post containing misinformation, it can spread to millions of people \cite{li2018influence}. Likewise, a small number of overloaded devices in the power grid can initiate a cascade that leaves thousands without power \cite{korkali2017reducing}. This potential for harmful effects to spread from a small initial subset of a network makes network environments an appealing target for malicious entities. An adversary will seek to maximize harmful impacts by simultaneously attacking multiple elements in the network. For example, in the electric power domain, a 2022 attack on two substations in North Carolina's electric grid were severely damaged by gunfire within a 10 minute interval, causing a cascade that left more than 45,000 people without power \cite{kerry2022northcarolina}. This combinatoric aspect of the attack creates a scalability problem when considering preemptive defense strategies, because as the size of the graph grows, the number of combinations of possible targets to prioritize grows exponentially for both the attacker and the defender.

\begin{figure}[h]
    \centering
    \includegraphics[width=\linewidth]{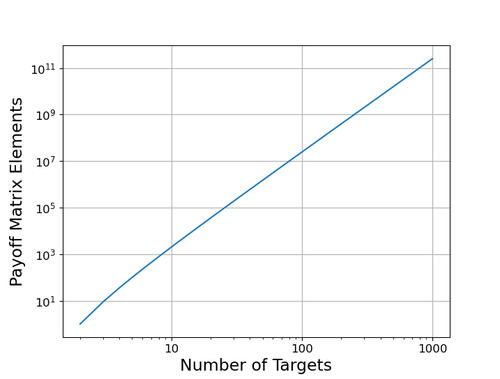}
    \caption{Size of payoff matrix for security game with $N$ targets and ${N}\choose{2}$ actions available to both the attacker and defender, plotted in logarithmic scale.}
    \label{fig:util_mat_size}
\end{figure}

While game-theoretic models provide solutions to the cascading impact defense problem in theory, in practice they have only been applied in a limited fashion due to the computational intractability of obtaining exact solutions for large networks with combinatorial actions \cite{wang2023attack}. As an illustrative example, in order to solve the exact Nash Equilibrium (NE) for a finite security game where the attacker and defender can each choose 2 out of $N$ targets, data for the cascading impact of every combination of possible attack and defense actions must be collected. This leads to a payoff matrix of size ${{N}\choose{2}}^2$, which grows exponentially as the number of targets increases, as shown in Figure \ref{fig:util_mat_size}. On top of this, in the general case, the complexity of solving for the NE is exponential with respect to size of the payoff matrix. This compounded exponentiality in complexity makes directly solving for the NE infeasible for all but the smallest networks when allowing combinatorial actions. As a result, a popular method for applying game theory to cascading impact defense is to restrict the number of possible targets to a subset of the elements in a graph, typically to less than $50$ combinatorial actions, for which directly solving the NE is still computationally viable \cite{paul2019learning,wei2016stochastic,guo2021reinforcement,yan2016q,ahmad2020maximizing}. A second State-of-the-Art (SOTA) approach is to reframe the possible actions from selecting individual targets to selecting from a set of predefined strategies, namely uniform random sampling and heuristic-based ranking \cite{li2019attack,chaoqi2021attack,wang2023attack}. By restricting the strategy space in this way, a separate NE can be computed that gives the optimal probabilities of each player selecting from these predefined strategies. This framing of the problem eliminates the exponential complexity that comes from large graphs and simultaneous attacks, as the number of actions is only proportional to the number of predefined strategies. However, this framing is much less expressive in terms of the possible strategies over all attack targets.%

Each of these two approaches reduce the complexity of the problem in order to accommodate traditional methods of solving for a NE. However, by employing a data-based approach to learn an approximate NE, the gap between these two SOTA approaches can be bridged, enabling solutions that are scalable to large networks and also fully expressive over the entire combinatorial action space. In this work, a Graph Neural Network (GNN) with multi-node embeddings is used to capture joint features between combinations of nodes and predict a cascading failure. This is combined with an action-restricting method in which many subsets of the full combiantorial action space are created. The GNN is then trained on a set of cascading impact outcomes in each of these sub-games and learns to generalize to the full combinatorial action space. We hypothesize that using training data that is structured as sub-games within the larger security game will improve generalization of the neural network compared to procuring the same amount of training data via sampling the entire action space.

However, while neural networks are able to generalize effectively when given enough data, collecting the volume of data they require is often a bottleneck. Real-world data for cascading impacts is typically very scarce, and even simulated data is often time consuming to produce depending on the fidelity of the simulation and the complexity of the cascade dynamics. In this work, Counterfactual Data Augmentation (CfDA) is also employed to overcome this bottleneck by generating additional data more efficiently than factual data can be obtained. CfDA leverages knowledge about the underlying dynamics of the domain to generate additional counterfactual training data from existing factual data that have been collected. It does so by combining aspects of different factual training examples, in this case nodes belonging to combinatorial actions, in new ways to create counterfactual data that is then efficiently validated as physically plausible. CfDA has shown success as a way to increase available training data for deep learning, even outside of combinatorial contexts \cite{pitis2020counterfactual,lu2020sample,buesing2018woulda}.

In order to make CfDA possible in the cascading impact domain, this work introduces algorithms for validating counterfactual data in the context of two generic models of cascade dynamics. For each of the cascading models covered, a counterfactual validation algorithm is presented that combines factual data collected from simulation in novel ways to create counterfactual data that can then be trained on by the neural network. Although there are many possible models of cascading impact, the two chosen for study in this work each address a broad class of cascade types and are generic to multiple application domains. In addition to works that use these cascading impact models as is, minor variations upon these models are ubiquitous in the cascading impact literature.

We show that the proposed method is able to generate close approximations to the NE strategies on 25 node graphs that are in the upper range of scale for which directly solving for the NE is feasible, as well as generate less exploitable strategies than the strategy-restricting methods on graphs as large as 1,000 nodes, for which directly solving for the NE is not feasible.

In summary, the main contributions of this work are as follows:
\begin{itemize}
    \item [$\bullet$] We introduce a method to learn arbitrary strategies over large graphs by partitioning the action space into subspaces during training and using GNNs with joint node embeddings.
    \item [$\bullet$] We introduce methods for the generation of counterfactual data from a factual set of data for two cascading impact models that are applicable to different sets of domains.
    \item [$\bullet$] We conduct extensive experiments that demonstrate the generalizability and scalability capabilities of the proposed method.
\end{itemize}

The rest of this paper is organized as follows; Section \ref{sec:lit_rev} reviews related literature on cascading impact modeling, graph environment security, and feature extraction from graphs. Section \ref{sec:prelim} details the game-theoretic formulation of the cascading impact security scenario, as well as the cascading impact models used in this work and the domains to which they are most relevant. Section \ref{sec:method} presents the algorithms for generating counterfactual data under each of the cascading impact models presented, as well as the details of the neural network architecture and training method. Section \ref{sec:results} details the experiments conducted to validate the proposed methods and provides a discussion on the results. In Section \ref{sec:conclusion}, the authors present conclusions drawn from this work and directions for future work.

\section{Related Work} \label{sec:lit_rev}
In this section, we review closely related work to the proposed method in graph feature extraction, graph environment security, and cascading impact models.

\begin{table*}[]
    \centering
    \caption{Comparison of the features of this work to those of the most closely related literature.}
    \begin{tabular}{ c|c|c|c|c|c|c}
      \textbf{Features} &  Guo et al. [1] & Chen et al. [2] & Kong et al. [3] & Wang et al. [4] & This Work
         \\
         \hline
      Large ($>100$ Targets) Networks &  & & \checkmark & \checkmark & \checkmark \\
      \hline
      Model Defensive Intervention & \checkmark & \checkmark & & \checkmark & \checkmark\\
      \hline
      Represent Arbitrary Strategy  & \checkmark & \checkmark &  &  & \checkmark\\
      \hline
      \end{tabular}
    \label{tab:litrev}
\end{table*}

\subsection{Graph Feature Extraction}
Graph Neural Networks (GNNs) have been widely used in tasks such as node classification, link prediction, and graph embedding on graph-structured data \cite{zhang2019graph}. K. Chen et al. \cite{chen2022gccad} developed GCCAD, an GNN designed for anomaly detection. This is a self-supervised framework that designs a graph corrupting strategy for generating synthetic node labels. This method being tuned for anomaly detection makes it an viable candidate for this problem, as nodes that a player may want to target are likely to be anomalous. B. Chen et al. \cite{chen2022distribution} introduce a method of graph embedding that removes the pooling of node features, which often reuslts in a critical loss of information, and instead models the distribution from which nodes are sampled. The proposed DKEPool outperforms state-of-the-art methods on graph classification tasks.

However, the typical way in which GNNs embed sets of multiple neighbors, by aggregating individual node embeddings, has severe limitations in that it cannot capture features that are dependent upon the joint set of nodes, such as the set of common neighbors between a set of nodes. This is problematic when interactions between specific nodes must be captured, as is the case with cascading impact and other domains such as gene co-expression networks and academic social networks \cite{xu2019network}. Thus, in this work we employ a multi-node embedding strategy similar to the \textit{labeling trick} proposed by Zhang et al. \cite{zhang2021labeling} to allow the neural network to capture these joint features between nodes.

\subsection{Security of Graph Environments}
Graph environment security is a complex and multi-faceted topic that has been studied from many different angles. In some contexts, such as the power grid, parts of the graph that have been disrupted by an attack can be recovered, and novel approaches for optimal failure recovery have been developed \cite{lu2018fast}. In other contexts, such as social media, methods have been developed to preserve user privacy from adversaries who seek to uncover private user data \cite{ding2019novel}. However, in this work, we limit ourselves to the preemptive defense of nodes in a graph environment, as this concept is nearly universal across application domains and particularly challenging in large graph environments. The preemptive defense problem is well-modeled using the game theoretic model of a security game, which is typically two-player zero-sum. In Section \ref{sec:introduction}, SOTA approaches for implementing the security game model in large graph environments were categorized into action-restricting methods and strategy-restricting methods.

Action-restricting methods identify a small subset of nodes or edges in the graph that are viable targets, effectively reducing the size of the security game. Wei et al. \cite{wei2016stochastic} use a secuirty game model in a power grid environment with $177$ targets and assume that the attacker and defender each choose two targets. The natural action space for this problem is ${177\choose2} \approx 15,000$, with with a payoff matrix size on the order of $10^8$. However, the authors employ action restriction to eliminate all but $40$ of these actions, and then, after simulating each of these actions, further reduce the action space to the $9$ that yielded positive utility to the attacker, for a payoff matrix of size $81$. Paul et al. \cite{paul2019learning} and Guo et al. \cite{guo2021reinforcement} both implement a Minimax Q learning approach to find the Nash equilibrium strategies for a graph with a total of 46 targets of which each player can select a single target at a time. Both of these works also manually prune the action space from $46$ for each player down to $10$ and $6$ respectively. Minimax Q learning is a well-known method for finding Nash equilibriums in two-player zero sum games that are too large to directly solve for the Nash equilibrium. Nonetheless, it still suffers from an inability to scale to large action spaces and, as these works demonstrate, requires manual pruning of the action space for large problems.

Strategy-restricting methods define a set of fixed strategies and reframe the actions in the security game as selecting from one of these predefined strategies as opposed to selecting a set of targets directly. Li et al. \cite{li2019attack} limit the strategy set to a uniform random strategy and a deterministic strategy that always chooses the nodes with highest degree, and calculate the corresponding NE. Chaoi et al. \cite{chaoqi2021attack}, use a similar approach with a different metric for the deterministic startegy, developing a customized cost function for the cascading dynamics. Wang et al. \cite{wang2023attack} use the Technique for Order Preference by Similarity to Ideal Solution (TOPSIS) criteria combined with the entropy weight method (EWM) to create their deterministic strategy. This method calculates an importance metric for each node by weighting four metrics of node centrality; namely degree centrality, closeness centrality, betweenness centrality, and eigenvalue centrality. EWM-TOPSIS is an algorithm that calculates the optimal weights for each of these centrality metrics to provide an accurate metric of the importance of the node relative to cascading impact. They also use uniform random node selection as the alternate strategy from which the NE is calculated.

Outside of the graph security domain, Deep RL (DRL) methods that leverage neural networks have demonstrated success in domains of multi-player games with very large state spaces, including the board games Go, Chess, and Shogie \cite{silver2016mastering,silver2017mastering}, as well as competitive video games such as DOTA2 \cite{berner2019dota} and StarCraft2 \cite{vinyals2019grandmaster}. These approaches use featurized state representations and neural networks to automatically learn to generalize knowledge about good action selections from states encountered during training to unseen states.

While action-restricting and strategy-restricting methods each allow for the security game model to be applied to large graph environments, they each have significant drawbacks either in scalability or expressiveness. Deep learning techniques for learning games have been able to scale to large state spaces while maintaining expressive policies, and this work uses a deep learning approach to achieve this feat in the large action spaces of the graph security problem.

\subsection{Models of Cascading Failure}
Cascading failure is a widely studied problem across a variety of domains. In this work, two models of cascading failure and their corresponding domains are focused on. The first is threshold-based methods of cascading. Under these models, a given node will fail if a certain threshold fraction of its neighbors fail. These methods have been applied to financial contagion as in Gai et al. \cite{gai2010contagion}, where an interconnected network of banks can have a chain reaction of failures if a bank fails that many other banks hold assets of. In the social network domain, influence maximization has been studied to maximize social influence by choosing a set of seed $k$ users in a social network. Influence maximization has received much attention due to its applicability to viral marketing, misinformation spread, and social recommendation \cite{li2018influence}. Borodin et al. \cite{borodin2010threshold} studied threshold models in the competitive influence of social networks.

The next cascading impact model studied in this work is the shortest path cascade that has been widely used in the transit domain. For instance Huange et al. \cite{huang2021using} used disaster spreading theory to analyze the cascading impact of an urban rail transit network, and \cite{ghena2014green} looked at similar cascading effects that can occur in networks of roads when a single road closure occurs to due an attack on traffic lights.

While are many different models of cascading failure, the two studied in this work were chosen for their generalizability to multiple domains given small variations. This choice was made so that the proposed CfDA generation algorithms could have as broad an applicability as possible to multiple application domains.

\section{Preliminaries} \label{sec:prelim}
In this section, important background information is provided for the analysis of this paper. Let us define a graph $G = (V, E)$, where $V = {v_1, v_2, \ldots, v_N}$ represents the set of $N$ nodes, and $E \subseteq V \times V$ represents the set of edges. Each node $v_i$ has an associated feature vector $\mathbf{x}_i \in \mathbb{R}^F$ that describes the attributes of the node.

Let us define a cascading failure model as $\Psi = C(\Theta)$, where $\Theta$ is the initial set of failed nodes due to an attack, and $\Psi$ is the set of failed nodes that failure has cascaded to after the system has reached steady state. The full set of failed nodes is then given by $\Omega = \Psi \cup \Theta$.

\subsection{Security Game Model}
We model the security scenario of preemptively defending a graph environment from an attacker who seeks to cause a cascading failure as a normal form security game, which has often been used to model the defense of infrastructure from adversaries \cite{chen2022game}. The game is defined by a finite set of actions available to both of the players and their corresponding payoffs. The game is zero-sum, $p_a = -p_d$, where $p_a$ is the payoff for the attacker and $p_d$ is the payoff for the defender.

In a normal-form security game, the defender and the attacker may employ mixed strategies, where they choose actions probabilistically. For an action space $A$ of size $|A|$, let $\mathbf{\pi}_d \in \mathbb{R}^{|A|}$ denote the probability that the defender will choose each action in $A$, and likewise $\mathbf{\pi}_a$ for the attacker. We denote the expected payoff with respect to a mixed strategy profile $(\mathbf{\pi}_a, \mathbf{\pi}_d)$ as $\mathbf{P}(\mathbf{\pi}_a,\mathbf{\pi}_d) = \mathbb{E}_{\mathbf{\pi}_a,\mathbf{\pi}_d} [p_a] = -\mathbb{E}_{\mathbf{\pi}_a,\mathbf{\pi}_d} [p_d]$

In the context of the combinatoric cascading failure problem, the set of possible actions for each player is a combination of targets to attack or defend. In this work, we limit ourselves to targeting $2$ nodes at a time, as this alone introduces scalability and strategy representation issues for large graphs. Thus, the action space for each player has size ${N\choose2}$ and each action ($\alpha_a$ for the attacker or $\alpha_d$ for the defender) represents a pair of nodes that the player selects. The total number of actions considering both players is clearly ${N\choose2}^2$, since every possible action for the attacker can be paired with every possible action for the defender.

The failure model used to derive the set of initial failures from the actions selected is given by:
\begin{equation}
    \Theta = \alpha_a - (\alpha_a \cap \alpha_d)
\end{equation}

Under this failure model, any nodes belonging to the attack action set fail as long as they are not also present in the defense action set. Once the final set of nodes $\Omega$ is computed from $\Theta$ via the cascading failure model, the payoff for each player is given by:

\begin{equation}
    p_a = -p_d = \frac{|\Omega|}{N}
\end{equation}
A strategy $\mathbf{b}_\pi$ is a best response to the strategy $\pi$ if and only if:
\begin{equation}
    \mathbf{b}_\pi = \max_{\pi'} \mathbf{P}(\pi',\pi)
\end{equation}
The expected payoff of the best response, $\mathbf{P}(\mathbf{b}_\pi,\pi)$, is referred to as the \textit{exploitability} of strategy $\pi$.

A Nash equilibrium in a normal form security game is a mixed strategy profile where both players are playing a best response to their opponents' strategy, and thus neither player has an incentive to unilaterally deviate from their chosen strategy. For a zero-sum attacker-defender security game, a Nash equilibrium is defined as:

\begin{equation}
\pi_a^* = \max_{\mathbf{\pi}_a} \min_{\mathbf{\pi}_d} \mathbf{P}(\mathbf{\pi}_a,\mathbf{\pi}_d)
\end{equation}

\begin{equation}
\pi_d^* = \min_{\mathbf{\pi}_d} \max_{\mathbf{\pi}_a} \mathbf{P}(\mathbf{\pi}_a,\mathbf{\pi}_d),
\end{equation}

\subsection{Cascading Failure}\label{sec:prelim_cascfailure}

In this study, we focus on simple cascading failure models that have been used in multiple domains. Among the numerous models available, we have specifically chosen two models for detailed analysis in this work. These models serve as representative examples of broad classes of cascade types that have been widely studied in the literature. Therefore, the selected models provide a solid foundation for demonstrating the generalizability of our method to multiple domains of graph security. In this section, we detail the properties of each of these two cascading failure models.

In general, a cascade consists of an initial set of failed elements that lead to a new set of failed elements that themselves lead to another set of failed elements, and so on until a steady state is achieved. Let us define $\Psi_t$ as the set of failed nodes caused by the $t$-th round of cascading such that $\Psi = \bigcup_{t=1}^T \Psi_t$, where $T$ is the number of rounds of cascading failure after which steady state is reached. Likewise, we define $\Omega^t$ as the total set of failed nodes as of the $t$-th round of cascading, with initial condition $\Omega_0 = \Theta$, such that $\Omega = \bigcup_{t=0}^T \Omega_t$. We can then define the function for a single-round of cascading as $\Psi_{t+1} = c(\Omega_t)$, of which the function $C(\cdot)$ is composed as follows:
\begin{equation}
    \Psi = C(\Theta) = \bigcup_{t=0}^{T-1} c(\Omega_t)
\end{equation}

Next, we detail the specific properties of the single-round cascade function $c(\cdot)$ for each of the cascading failure models studied in this work.

\subsubsection{Threshold-based Cascading Failure}\label{sec:prelim_threshcasc}
We define a set of thresholds such that each node $v \in V$ has a corresponding threshold $\phi_v \in (0,1]$. These thresholds can be assigned to nodes using domain knowledge or empirical data. To create simulated data, thresholds can be generated according to some probabilistic distribution. The threshold-based cascading model is defined by the following single-round cascade function:
\begin{equation}\label{eq:threshcasc}
 c(\Omega_t) =  \forall v \notin \Omega_t \quad \text{such that} \quad \frac{|\mathcal{N}(v) \cap \Omega_t| }{|\mathcal{N}(v)|} \geq \phi_v
\end{equation}
where $\mathcal{N}(v)$ denotes the set of neighbors of node $v$.

\begin{figure}[h]
    \centering
    \includegraphics[width=0.75 \linewidth]{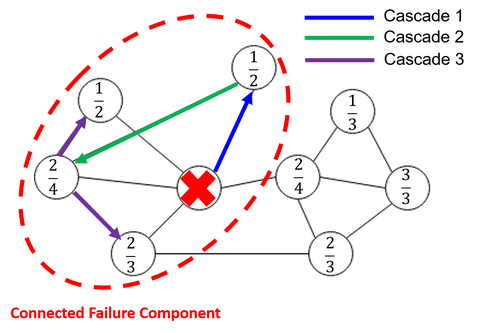}
    \caption{An example of a cascade in the threshold-based cascading model. The initial failed node is marked with a red "X", and this failure causes Cascade 1 which is marked with a causal arrow. As a result of Cascade 1, Cascade 2 occurs, and so on with Cascade 3.}
    \label{fig:threshexample}
\end{figure}

Figure \ref{fig:threshexample} shows an illustrative example of a cascade under the threshold-based model. Cascade 1 is triggered by the initial failure, as one of its neighbors has a threshold of 1/2 and only 2 neighbors, thus meeting the criteria for a cascading failure. Once both of these nodes have failed after Cascade 1, this triggers another node to fail in Cascade 2 as it has a threshold of 2/4 and is neighbors with both failed nodes. Cascade 3 follows from Cascade 2, with two additional nodes meeting their thresholds for cascading failure from the set of failed nodes in Cascades 1 and 2.

An important consequence of (\ref{eq:threshcasc}) and the constraint of $\phi_j > 0$ is that all nodes in $\Psi_{t+1}$ must have at least one neighbor in $\Omega_t$. Because $\Omega_0 = \Theta$, all subsequent cascades will contribute to growing connected components of failed nodes originating from one or more of the initial failures as can be seen in the example of Figure \ref{fig:threshexample}. Variations in threshold-based models have been used to model cascading failure in social and financial networks \cite{borodin2010threshold,gai2010contagion}.

\subsubsection{Shortest-Path Cascading Model}
In this cascading failure model, an exchange of some quantity (e.g., traffic, energy, information) is assumed between all pairs of nodes along the shortest path connecting two nodes. Each node $v \in V$ carries a load $l_v$ that is the number of shortest paths between all nodes that pass through $v$, not including the trivial shortest paths for which $v$ is an endpoint (without loss of generality). Each node also has a capacity $k_v$, which similar to $\phi_v$ in the threshold model is assigned based on either empirical data, or according to a probabilistic distribution with the constraint that $k_v > \lambda_v$, where $\lambda_v$ is the load on node $v$ under the initial conditions of the graph with no node failures. When a node $v$ fails, its load will be distributed to other nodes according to the new shortest paths that are created in the absence of $v$ and other nodes that may have failed. Given the failure set $\Omega_t$, we define the function to calculate the load on a particular node $v$ as $l_v^t = \ell_v(\Omega_t)$. The single-round cascading failure function for shortest-path cascading is as follows:

\begin{equation}\label{eq:shortestpath_casc}
    c(\Omega_t) = \forall v \notin \Omega_t \quad \text{such that} \quad l_v^t > k_v
\end{equation}

Figure \ref{fig:shortpath_example} shows an illustrative example of this cascading failure model. The node that fails initially has a load of $6$, as it is part $6$ shortest paths between other nodes. When this node fails, these shortest paths are routed through other nodes that gain $6$ load as a result.

\begin{figure}[h]
    \centering
    \includegraphics[width=0.75\linewidth]{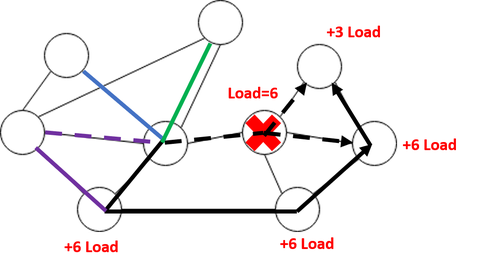}
    \caption{An example of load redistribution according to shortest-path cascade rules. The red ``X" marks the initial failed node, and the dashed lines represent the shortest paths between nodes that traveled through this node before it failed. The solid lines represent the redirected shortest paths after the node failed. The black lines indicate that multiple paths are following the same route. Also indicated is all the nodes that gain load as a result of the initial failure, and how much load is gained.}
    \label{fig:shortpath_example}
\end{figure}

Variations in shortest-path models have been used to model cascading failure in traffic and network routing applications \cite{zeng2022transportation}.

\section{Method}\label{sec:method}

In this section, the proposed method is detailed, including the generation of the factual training data, the generation of the counterfactual training data, and the model architecture used to learn to predict cascading failure outcomes.

\subsection{Factual Data Generation}\label{sec:fac_gen}
In order to generate the training dataset, the action space $A$ of size $|A| = {N\choose2}$ is partitioned into $p$ subaction spaces $\{\ddot{A}_i\}_0^{p-1}$each containing $M$ nodes and ${M\choose2}$ actions. Then, $q$ trials are performed for each subaction space, where a trial consists of simulating the cascading failure for a particular choice of actions of the attacker and defender that belong to the subaction space. When $M$ is small, $q$ can be chosen to be ${M\choose2}^2$, so that all possible actions can be exhausted in the subaction space. Figure \ref{fig:facdata_method} shows an illustration of the process of generating factual data. %

\begin{figure}[h]
    \centering
    \includegraphics[width=1.02 \linewidth]{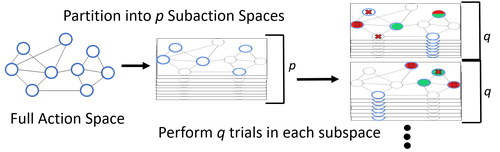}
    \caption{Illustration of the factual data generation process. The full action space with all nodes available to both players is broken down into $p$ subaction spaces with limited sets of nodes available to target (outlined in blue). In each of the $q$ trials, the node choices by the attacker (red) and defender (green) are recorded as well as the nodes that belong to $\Omega$ for the trial (marked with red ``X").}
    \label{fig:facdata_method}
\end{figure}

After these trials have been performed, the result is a dataset $\mathcal{D}(\boldsymbol{\alpha}_a,\boldsymbol{\alpha}_d,\mathbf{\Omega})$, where $\boldsymbol{\alpha}_a \in \mathbb{R}^{pq \times 2}$ is the set of all node pairs chosen by the attacker for all $q$ trials in all $p$ subaction spaces, $\boldsymbol{\alpha}_d \in \mathbb{R}^{kp \times 2}$ is the same for the defender, and $\mathbf{\Omega} \in \mathbb{R}^{kp \times N}$ is a multi-hot encoding of the final set of failed nodes after cascading failure has completed given the corresponding attack and defense actions, where a value of $1$ indicates that that particular node in the network failed as a result of cascade in that trial.

\subsection{Counterfactual Data Generation} \label{sec:cf_gen}
CfDA complements the proposed action space decomposition by combining node selections by both players made in different subaction spaces and their corresponding cascades to create novel scenarios that are not represented in $\mathcal{D}$. In this way, CfDA adds data that crosses over between the subaction spaces. However, because these are novel scenarios with respect to the factual data, the combined counterfactual cascades must be validated as physically plausible given the counterfactual action selections. Figure \ref{fig:CfacMethod} shows an example of the proposed CfDA process.

To validate the combination of actions from different trials, it must be possible to isolate the contribution of a specific attacked node $v \in \Theta$ to the final cascade $\Omega$. This attribution process is specific to each cascading failure model, and trials for which it is not possible are discarded for CfDA. Let us denote the contribution to $\Omega$ of the initial failure of node $v \in \Theta$ as $\Omega_{v}$. Note that $\Omega_v \subseteq \Omega$.

A counterfactual trial will begin with selecting nodes $v \in \Theta_1$ and $w \in \Theta_2$ from two factual trials from different subaction spaces, to create the counterfactual initial failure set $\hat{\Theta} = \{ v,w\}$. Then, the counterfactual final cascade is constructed as $\hat{\Omega} = \Omega_v \cup \Omega_w$. $\hat{\Omega}$ is considered a valid counterfactual scenario under the following condition:
\begin{equation}\label{eq:val_criterion}
    c(\hat{\Omega}) = \emptyset
\end{equation}
Note that this condition is dependent on the cascading failure dynamics, and requires that $\hat{\Omega}$ is at steady-state. In order for CfDA to be beneficial, it is important that the condition (\ref{eq:val_criterion}) can be evaluated more efficiently than evaluating the factual outcome of $\hat{\Omega} = C(\hat{\Theta})$. This section presents novel methods to perform this efficient validation of counterfactual data for multiple models of cascading failure.

\begin{figure}
    \centering
    \includegraphics[width=1.38\linewidth]{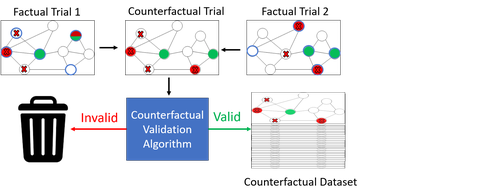}
    \caption{Illustration of the general counterfactual generation process. A counterfactual action is created by mixing attack and defense actions from two different factual trials that come from different subaction spaces. All failures associated with the attack actions are also incorporated into the counterfactual trial. Then, this counterfactual example is evaluated on its realizability under the specified cascading dynamics and added to the counterfactual dataset if it is feasible.}
    \label{fig:CfacMethod}
\end{figure}

\subsubsection{Threshold-based Cascading Model}

Recall the property of threshold-based cascading in Section \ref{sec:prelim_threshcasc} which states that $\Omega$ can be represented as a set of connected components containing one or more initial failures each. Because of this property, the attribution of $\Omega_v$ can be done by identifying the connected component $\Gamma_v \in \Omega$ such that $v \in \Gamma_v$. If there are two nodes in $\Theta$ but only a single connected component, then this indicates that the cascades originating from these two initial failures have merged, and this trial is discarded for the purposes of CfDA. Otherwise, we take $\hat{\Omega} = \Gamma_v \cup \Gamma_w$, where $v \in \Theta_1, w \in \Theta_2$.

\newtheorem{theorem}{Proposition}
\begin{theorem}\label{prop:thresh}
Given $c(\cdot)$ is the threshold-based model of cascading failure as defined in (\ref{eq:threshcasc}), the condition (\ref{eq:val_criterion}) is true if and only if $\phi_v > \frac{|\mathcal{N}(v) \cap \hat{\Omega}|}{|\mathcal{N}(v)|}, \; \forall v \notin \hat{\Omega}$ holds.
\end{theorem}
\begin{proof}
Take $\hat{\Omega}$ to be $\hat{\Omega}_0$. Condition (\ref{eq:val_criterion}) states that $\hat{\Omega}_1 = \emptyset$. According to (\ref{eq:threshcasc}) a node $v \in \hat{\Omega}_1$ if and only if $\phi_v \leq  \frac{|\mathcal{N}(v) \cap \hat{\Omega}_0|}{|\mathcal{N}(v)|}$ and $v \notin \hat{\Omega}_0$. Therefore, by contradiction, $\hat{\Omega}_1 = \emptyset$ if and only if $\phi_v > \frac{|\mathcal{N}(v) \cap \hat{\Omega}_0|}{|\mathcal{N}(v)|}, \; \forall v \notin \hat{\Omega}_0$.
\end{proof}

Proposition \ref{prop:thresh} can be evaluated more efficiently than calculating the factual sample $\hat{\Omega} = C(\hat{\Theta})$, as it it only performs a single round of cascading failure computation, as opposed to the $T$ steps of computation that is required in calculating a factual sample. Intuitively, this counterfactual data allows for many round of cascading computations to be reused in other scenarios where the same cascades would also occur.

\subsubsection{Shortest-Path Cascading Model}
Unlike the threshold-based model, the shortest-path model has no requirements that $\Omega_v, v \in \Theta$ be a connected component. %
Rather, the failure of $v$ will induce some change in the load on all other nodes in the graph, which initiates the cascading failure process. Let us define the vector $\boldsymbol{\gamma}_v \in \mathbb{R}^{N-1}$ as:
\begin{equation}
    \boldsymbol{\gamma}_v = \ell_w(\Omega_v) - \lambda_w, \quad \forall w \in V, w \neq v
\end{equation}

This vector $\boldsymbol{\gamma}_v$ represents the change in load on all other nodes as a result of the failure of $v$. Note that because $l_v$ represents the number of non-trivial shortest paths that pass through $v$, no other node can gain more than $l_v$ in load due to the failure of $v$. Therefore, the load gained by any other node $w$ as a result of the failure of node $v$ can be upper bounded by:
\begin{equation}\label{eq:shortpath_ub}
    \boldsymbol{\gamma}_{v}(w) \leq \sum_{z \in C(v)} l_z, \quad w \neq v
\end{equation}
This inequality represents the worst-case-scenario from the perspective of node $w$, that all shortest-paths that were routed through nodes in $\Omega(v)$ are now routed through $w$ as a result of its failure.

\begin{theorem}\label{prop:shortpath}
Given $c(\cdot)$ is the shortest-path model of cascading failure as defined in (\ref{eq:shortestpath_casc}), the condition (\ref{eq:val_criterion}) is true if for $v \in \Theta_1$ and $w \in \Theta_2$ either of the following two conditions hold:\\
$k_n \geq
\lambda_n + \boldsymbol{\gamma}_v(n) + \sum_{z \in C(w)} l_z, \; \forall n \notin \hat{\Omega}$ \\
$k_n \geq \lambda_n + \boldsymbol{\gamma}_w(n) + \sum_{z \in C(v)} l_z, \; \forall n \notin \hat{\Omega}$
\end{theorem}
\begin{proof}
Take $\hat{\Omega}$ to be $\hat{\Omega}_0$. Condition (\ref{eq:val_criterion}) states that $\hat{\Omega}_1 = \emptyset$. According to (\ref{eq:shortestpath_casc}), $n$ only belongs to $\hat{\Omega}_1$ if $k_n < \ell_n(\hat{\Omega}_0)$. We have $\ell_n(\hat{\Omega}_0) = \ell_n(\Omega(v)) + \ell_n(\Omega'(w))$, where $\Omega'(W)$ indicates that the cascade is performed on the modified graph with nodes in $\Omega(v)$ already removed. For this reason, while we can take $\ell_n(\Omega(v)) = \boldsymbol{\gamma}_v(n)$, it is not necessarily true that $\ell_n(\Omega'(w)) = \boldsymbol{l}_w(n)$. However, the upper bound in (\ref{eq:shortpath_ub}) can be substituted to create the lower bound $k_n \geq l_n^I + \Delta \boldsymbol{\gamma}_v(n) + \sum_{z \in \Omega(w)} l_z, \; \forall n \notin \hat{\Omega}$. The subscripts $v$ and $w$ can be trivially switched to create the second lower bound.
\end{proof}

The conditions of Proposition \ref{prop:shortpath} can be efficiently evaluated by first calculating $\Delta \mathbf{l}_v, \quad \forall v \in V$, as well as the upper bound (\ref{eq:shortpath_ub}) for all nodes. The number of calculations involved in calculating these quantities in used in Proposition \ref{prop:shortpath} scales linearly with $N$, making it efficient to calculate in an action space that grows exponentially.

\subsection{Payoff Prediction Model} \label{sec:pp_model}

In this section, the proposed method is detailed for using the dataset created as outlined Sections \ref{sec:fac_gen} and \ref{sec:cf_gen} to train a neural network model to predict $\Omega$ for a trial given attacker and defender actions $\alpha_a$ and $\alpha_d$. This prediction model is performing multi-label binary classification using the multi-hot encoding of $\Omega$ as a label. The model aims to predict the likelihood of each node in the network failing for a given trial. The goal of this method is to train this neural network on a small fraction of the exponentially large number of actions, such that it can accurately predict actions that it has not seen before. Then, the neural network's capability for efficient parallel evaluation will allow it to predict outcomes for the rest of the actions for which there is no data. Subsequently, strategies can be generated based on these predictions.

Under the cascading failure models studied in this work, each node has the feature of either a threshold or a capacity that impacts its cascading dynamics. Thus, the node feature matrix that is passed to the model $\mathbf{X}\in\mathbb{R}^{N}$ consists of the thresholds or capacities of each node.

Because binary input features are not conducive to learning for neural networks, the multi-hot encodings are mapped to a continueous embedding space. Each of the ${N\choose2}$ actions is mapped to a unique continuous embedding, and learnable parameters aim to map similar encodings to similar continuous embeddings. Note that this step is a multinode embedding, as each \textit{pair} of nodes is embedded jointly. This is contrary to most GNN node embedding techniques that aggegregate single-node embeddings, and ignore joint features of node pairs. In this context where actions correspond to node pairs, it is critical to capture these joint features, as the cascading effects of a particular node pair cannot always be captured by considering the individual properties of each node alone. This joint embedding of multihot encodings is similar to the \textit{labeling trick} proposed by Zhang et al. \cite{zhang2021labeling} to learn multinode features.

\begin{figure}[h]
    \centering
    \includegraphics[width=1.02\linewidth]{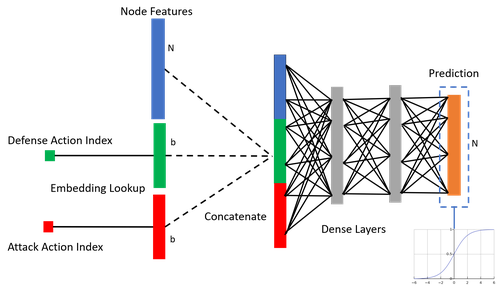}
    \caption{The neural network model used to predict cascading failure payoffs in our proposed method. Takes as input the embedding of the overall graph from the graph feature extractor, as well as the node embeddings of the attacker's and defender's node choices and returns the payoff value as output. Dimensions of each component of the architecture are shown on the figure, where $b$ is the number of features in the embedding and $h$ is the hidden layer size. }
    \label{fig:q_critic}
\end{figure}

The node pair embeddings for both the attack and defense actions are then concatenated with the threshold/capacity node feature vector to create a learned encoding of the feature vector. This encoding is then passed through a series of fully connected layers of fixed size $h$, before the final layer maps it to dimension $N$. A sigmoid activation function is used on the output layers so that all values are in the interval $(0,1)$ and correspond to the predicted probability that each node is part of the cascading failure. Figure \ref{fig:q_critic} shows an illustration of the overall architecture.

\section{Experiments} \label{sec:results}
In this section, extensive experiments are conducted to validate the proposed method and benchmark it against the most comparable approaches. The implementation of the proposed method used in these experiments can be found on the project's GitHub repository.\footnote{https://github.com/jdc5549/cfda-network-envs}

\subsection{Factual Data Generation}
We choose to examine test-cases of graphs with $25,100,$ and $1000$ nodes. The 25-node test-case evaluates the proposed method's ability to approach NE solutions on problems for which the NE can be solved. A security game of this size approaches the upper range for which it is reasonable to directly calculate the NE. The size of the utility matrix for a 25-node graph where two simultaneous targets can be chosen is $|U| = {25\choose2}^2 = 90,000$. The $100$ and $1000$ node test cases evaluate the proposed method in scenarios where directly solving for the NE is not feasible.

For each of these test cases, the subaction sets were chosen to have $5$ targets available to each player, for a total of $100$ possible combinatorial actions considering both players. This choice allows every combinatorial action in each subspace to be fully explored by choosing the number of trials $k=100$. The number of subaction sets $p$ varied for each test-case, and was tuned for best performance.

After the dataset $\mathcal{D}$ is generated, duplicate combinatorial actions are filtered out, resulting in approximately $4\%$ of data being filtered in each test case.

\subsection{Counterfactual Data Generation}

As discussed in Section \ref{sec:cf_gen}, counterfactual data is used to fill in gaps between subaction sets, as not all combinatorial actions are represented in the subaction sets. Thus, a natural question to ask is how the number of subaction sets impacts the amount of counterfactual data that can be generated. Table \ref{tab:subact_vary} shows the number of counterfactuals that are generated for varying amounts of subaction sets in a 25 choose 2 security game with threshold-based cascading. It should be noted that the full payoff matrix for this game contains 24.5 million unique actions. We can see that there is a substantial increase in the proportion of counterfactual to factual data when there are twice as many subsets as nodes, and this ratio stays relatively stable stable as the number of subaction sets increases. Moreover, the efficiency of counterfactual generation is higher than the efficiency of factual data regardless of the number of subaction sets used. For the remaining experiments in this section, the number of subaction sets is chosen to be triple the number of nodes in the network, as this seems to strike the best balance between number of counterfactuals generated and efficiency.

\begin{table}[h]
    \centering
    \begin{tabular}{|c|c|c|}
        \hline
        Num. Sets & Num. Fac. (ms/data) & Num. CFac (ms/data)\\
        \hline
       25 & 2490 (4.02) & 4840 (1.69)\\

       \hline
       50  & 4980 (3.69) & 28,000 (1.86) \\
       \hline
       75 & 7440 (3.80) & 59,200 (1.72) \\
       \hline
       100 &  9880 (3.45) & 78,300 (1.88)\\
       \hline
    \end{tabular}
    \caption{Number of counterfactuals generated and efficiency of generation (in ms/data) for varying number of subaction sets of size 5 in a 100-node network with threshold-based cascading}
    \label{tab:subact_vary}
\end{table}

Figure \ref{fig:cfac_gen} shows the number of counterfactuals that can be generated from a given number of factual samples in a 100-nodes threshold-based cascading environment. The figure shows roughly quadratic growth in the number of generated counterfactuals, which is sometimes referred to as the ``blessing of dimensionality". However, some of the experiments conducted in this work generate a large amount of factual data, such that generating all of the counterfactual data that is possible is both impractical and unnecessary. Thus, going forward, counterfactual data generation is artificially capped at a factor of $10$ times the number of factual data generated.

\begin{figure}
    \centering
    \includegraphics[width=1.04\linewidth]{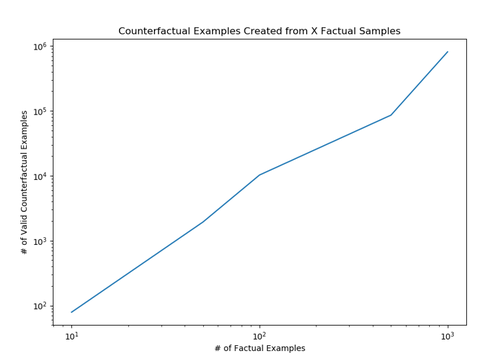}
    \caption{Number of Counterfactual examples that can be generated from a given number of factual examples in a 100 node threshold-based security game environment.}
    \label{fig:cfac_gen}
\end{figure}

Bearing this in mind, Table \ref{tab:num_cfac_data} shows the number of counterfactuals that were generated for each graph size, with $3N$ subaction sets each. We can see that across all of the test cases, the threshold-based cascading counterfactual model generates slightly more data than the shortest-path-based cascading model, however both are able to consistently generate more than the amount of factual data available. Furthermore, in the 100 and 1000 node test cases, the number of counterfactuals generated is very close to the artificial cap of $10$
times the amount of factual data.

\begin{table}[h]
    \centering
    \begin{tabular}{|c|c|c|c|}
    \hline
        Graph Nodes & Casc. Type & Factual Data & Counterfactual Data\\
        \hline
       \multirow{2}{*}{25} & Threshold & $6.75 \times 10^{3}$& $1.15 \times 10^{4}$  \\
        \cline{2-4} & Shortest Path & $6.63 \times 10^{3}$  & $9.76 \times 10^3$\\
        \hline
        \multirow{2}{*}{100} & Threshold & $3.09 \times 10^4$ & $1.60 \times 10^5$ \\
        \cline{2-4} & Shortest Path &  $3.10 \times 10^4 $& $1.23 \times 10^5$\\
        \hline
        \multirow{2}{*}{1000} & Threshold &  $3.19 \times 10^{5}$ & $3.81 \times 10^6$\\
        \cline{2-4} & Shortest Path & $3.19 \times 10^{5}$ & $2.95 \times 10^6$\\
        \hline
    \end{tabular}
    \caption{KL Divergence to the pre-caclulated NE and time elapsed for data creation and training (if applicable) for each method. }
    \label{tab:num_cfac_data}
\end{table}

These results demonstrate that the proposed counterfactual generation methods are successful in generating a vast quantity of data more efficiently than directly acquiring factual data, under multiple models of cascading.

\subsection{Security Game Experiments}

In order to validate the merit of the proposed method in the context of cascading impact defense, we benchmark against SOTA action-restricting and strategy-restricting methods. For action-restricting benchmarks, we choose We et al. \cite{chen2022game} as a method that directly solves for the NE via linear programming and Guo et al. \cite{guo2021reinforcement} as a method that employs tabular Minimax Q-learning to learn the NE strategies. While both of these methods are expected to achieve precise approximations of the NE, they only are scalable to the smallest of scenarios presented in this work, i.e. a 25-node scenario. Therefore, they are not included as benchmarks for larger scenarios. For all test cases, the strategy restriction approach employed by Wang et al. \cite{wang2023attack} is used as a benchmark. This work follows the strategy restriction framework detailed in Section \ref{sec:lit_rev}, where a NE mixed strategies between the selection of nodes uniformly and deterministically according to a heuristic is solved.

We also evaluate ablations of key aspects of the proposed method, namely the partitioning of the action space into subspaces with restricted target sets, as well as CfDA. We evaluate the neural network model without either of these aspects by acquiring training data by randomly sampling the combinatorial action space until $qp$ trials have been collected. We also evaluate the model with action space partitioning while training only on the factual training data, to evaluate the importance of CfDA. We do not evaluate CfDA in the unpartitioned action space, as the motivation for including CfDA is to fill the gaps between subaction spaces.

The 25-node test case evaluates the proposed method in a scenario for which it is feasible to calculate the NE. Table \ref{tab:25node_results} shows the Kullback-Leibler divergence to ground truth NE between the proposed method and the benchmark methods of direct NE computation and Minimax Q-learning, as well as the strategy-restricted method that limits to a mixture of uniform random and deterministic targeted by node degree strategies. Also included are the aforementioned ablations on components of the proposed method. The table also shows the time of computation, which for the neural network methods is based on the time to generate the data and to perform training. The best performing methods for each cascading impact dynamics model in both the scalable and non-scalable categories have their results bolded.

\begin{table*}[h]
    \centering
    \begin{tabular}{|c|c|c|c|c|}
    \hline
        Method & Casc. Type & KL(NE)  & Time (min) & Scalable\\
        \hline
       \multirow{2}{*}{Chen et al. (2022) \cite{chen2022game}} & Threshold & \textbf{0} & 10 &\\
        \cline{2-4} & Shortest Path & \textbf{0} & 8 & No\\
        \hline
        \multirow{2}{*}{Guo et al. (2021) \cite{guo2020enhancing}} & Threshold &$6.09 \times 10^{-3}$ & 42 & \\
        \cline{2-4} & Shortest Path & $4.49 \times 10^{-3}$& 39 & No\\
        \hline
        \multirow{2}{*}{Wang et al. (2023) \cite{wang2023attack}} & Threshold & 3.51 & $\mathbf{0.17}$ &\\
        \cline{2-4} & Shortest Path & 5.41 & $\mathbf{0.18}$ & Yes\\
        \hline
        \multirow{2}{*}{NN} & Threshold & 5.34 & 20 &\\
        \cline{2-4} & Shortest Path & 7.18 & 10 & Yes\\
        \hline
        \multirow{2}{*}{NN+Subact} & Threshold & $\mathbf{1.68}$ & 21 & \\
        \cline{2-4} & Shortest Path & $\mathbf{7.06 \times{10^{-5}}}$ & 18 & Yes\\
        \hline
         \multirow{2}{*}{\textbf{(Ours) NN+Subact+CfDA}} & Threshold & 1.85  &  47 &\\
        \cline{2-4} & Shortest Path & $1.45 \times 10^{-2}$ & 42 & Yes\\
        \hline
    \end{tabular}
    \caption{KL Divergence to the pre-caclulated NE and time elapsed for data creation and training (if applicable) for each method. }
    \label{tab:25node_results}
\end{table*}

Note that each of these benchmarks provides an extreme either in accuracy, as with direct NE computation, or in efficiency, as with the strategy-restricted method. However, we can see that while the strategy-restricted method can provide a strategy in little time, it is very far from the NE. The action-restricting methods are the most appealing for a network for which it is feasible, as expected, but it will cease to be viable to compute for the large network scenarios that are considered in this work.

For all test-cases, including the 100 and 1000 node cases in which action-restricting methods are not feasible, the exploitability and validation set error were also used as evaluation metrics. The validation error metric applies only to the neural network models, as the strategy-restricted methods do not directly make predictions about the failure set $V_F$. The validation error is calculated with respect to a validation dataset of size $10,000$ that has no overlap with the training dataset. This evaluates the generalizability of the neural network approaches in being able to accurately predict the outcome of a cascading failure scenario that it has not seen before. The exploitability of a strategy is defined as the expected average payoff of the best response to that strategy. An exploitability of $2\delta$ yields at least a $\delta$-Nash equilibrium \cite{heinrich2016deep}, which means that the Nash equilibrium has an exploitability of $0$ by definition. We can approximate this best response strategy by training new agents specifically against the fixed ``ego" agent models we wish to evaluate. Because these ``exploiter" agents always play against a fixed ego agent strategy, this is a single-agent game rather than a multi-agent game, and the size of the action space does not grow exponentially. We use Multi-Armed Bandit Reinforcement Learning (MABRL) to train exploiter agents to maximize their payoff in this single-player game.

More precisely, if we have trained ego agents $\pi_{\text{Atk}}$ and $\pi_{\text{Def}}$ that have average payoffs playing against each other of $\delta_{g}$ and $-\delta_{g}$ respectively, we would train exploiter $\pi_{\text{XA}}$ against $\pi_{\text{Atk}}$ and $\pi_{\text{XD}}$ against $\pi_{\text{Def}}$ using MABRL, and they would achieve average payoffs of $\delta_{XA}$ and $\delta_{XD}$ respectively. Our calculated exploitability would then be
\begin{equation}\label{eq:explt}
    \delta = (\delta_{XD} - \delta_g) + (\delta_{XA} + \delta_g)
\end{equation}

\begin{table*}[h]\label{tab:expl_val_results}
    \centering
    \begin{tabular}{c|c|c|c}
       Graph nodes & Method & Exploitability  & Validation Err\\
       \hline\hline
      \multirow{4}{*}{\textbf{25}}  & Wang et al. (2023) & $0.111$ & N/A \\
        \cline{2-4}
        & NN & $0.0102$ & $0.0151$\\
        \cline{2-4}
        & NN+Subact &  $6.92 \times 10^{-3}$  &  $4.68 \times 10^{-3}$\\
        \cline{2-4}
        & \textbf{(Ours) NN+Subact+CfDA} & $\mathbf{5.56 \times 10^{-3}}$ & $\mathbf{4.20 \times 10^{-3}}$\\
        \hline\hline
      \multirow{4}{*}{\textbf{100}}  & Wang et al. (2023) & $0.0715$ & N/A\\
        \cline{2-4}
        & NN & $0.0728$ &  $4.81 \times 10^{-3}$ \\
        \cline{2-4}
        & {NN+Subact} & $0.0675$ & $5.09 \times 10^{-3}$\\
        \cline{2-4}
        & \textbf{(Ours) NN+Subact+CfDA} & $\mathbf{0.0642}$ & $\mathbf{2.24 \times 10^{-3}}$  \\
        \hline \hline
      \multirow{4}{*}{\textbf{1000}}  & Wang et al. (2023) & $5.88 \times {10}^{-3}$ & N/A\\
        \cline{2-4}
        & NN & $6.13 \times 10^{-3}$  & $2.19 \times {10}^{-3}$\\
        \cline{2-4}
        & {NN+Subact} & $6.04 \times 10^{-3}$  &  $2.18 \times 10^{-3}$ \\
        \cline{2-4}
        & \textbf{(Ours) NN+Subact+CfDA} & $\mathbf{5.14 \times 10^{-3}}$ & $\mathbf{1.94 \times 10^{-3}}$ \\
    \end{tabular}
    \caption{Exploitability and validation error for proposed method and ablations as well as SOTA strategy-restricted benchmark.}
    \label{tab:1000node_results}
\end{table*}

Table \ref{tab:1000node_results} shows the exploitability and validation error results for all of the test-cases and benchmark methods. In Table \ref{tab:25node_results}, we can see that the addition of subaction spaces to the neural network model resulted in a drastic improvement in the KL divergence to the NE in 25-node networks, while including counterfactual data slightly increases the divergence to the NE. However, in Table \ref{tab:1000node_results}, we see that the inclusion of counterfactual data does result in modest improvement in terms of validation error and exploitability. This discrepancy is likely due to the NE divergence tending to be a noisy evaluation metric combined with the similar performance in prediction accuracy of these two approaches in the 25-node test case.

As the size of the network increases to $100$ and $1000$ nodes in Table \ref{tab:1000node_results}, we can see that the absolute difference in both exploitability and validation error between of all approaches decreases. This is due to the fact that as the size of the network increases, while the number of targets being selected by each player remains at $2$. Thus, the impact and corresponding payoff for each player is a smaller fraction of the overall network size, which is correlated with both the exploitability and validation error. So while the absolute difference between the proposed method and the others decreases with the size of the network, the relative gap widens in favor of the proposed method as the network size increases. This is the expected result considering the results from Table \ref{tab:num_cfac_data}, as the amount of ``missing" factual data for which the counterfactual data compensates is smallest in the 25-node test case, and grows exponentially with the size of the network. Thus, the $1000$-node test benefits the most from the generation of counterfactual data, and this is reflected in the improved prediction accuracy and exploitability.

Moreover, the SOTA strategy-restricting method is clearly outperformed by all of the neural network approaches in terms of exploitability. While the strategy-restricting approach is computationally efficient, the lack of expressive power in representing player strategies clearly leaves room for flexible neural network models to generate superior strategies.

\section{Conclusion} \label{sec:conclusion}
In conclusion, this paper addresses the challenge of generating preemptive defense strategies against intelligent adversaries in large graph environments. We propose partitioning the combiantorial action space into smaller subaction spaces to create a dataset of cascade outcomes, and augmenting this dataset with counterfactual data. We then train a neural network to predict cascading failure outcomes based on this data, and generalize to any combination of targets chosen by the attacker and defender.

The experimental results demonstrate the efficacy of the proposed method. It achieves significant reductions in exploitability compared to state-of-the-art (SOTA) methods when applied to large graph scenarios. Additionally, the method produces strategies that closely approximate the Nash Equilibrium in smaller-scale scenarios where computation of the equilibrium is feasible. Moreover, the inclusion of subaction space partitioning is shown to provide significant gains in prediction accuracy compared to collecting a similar amount of training data by randomly sampling the full combinatorial action space.

While this work demonstrates the potential of deep learning techniques to generate superior strategies in large combinatorial action spaces, there are many directions for expansion upon this work. In this work, two computationally simple types of cascading failure are investigated and counterfactual data generation algorithms are proposed that are informed by these dynamics. For more computationally complex models of cascading failure with higher fidelity, the benefit of counterfactual data is even greater, but validation algorithms are more challenging to develop. A learning-based solution to counterfactual validation via physics-informed neural networks is an appealing option for a model that is generalizable to any cascading failure dynamics.

Another possible direction of future work is more rigorous process of choosing how to partition subaction spaces. In this work, subaction spaces were sampled randomly creating overlap between combinatorial actions and only probabalistically guaranteeing the representation of every node. Mathematical concepts from combinatorics such as Steiner Systems could be investigated as a way to optimally partition the combinatorial action space.

Additionally, alternative game theoretic formulations of the problem are compatible with this method. While the symmetric zero-sum security game was an appealing choice to demonstrate the potential of this method due to its simplicity, in practice there is often an imbalance in resources between the attacker and defender, making the game asymmetric. Moreover, more complex models of resource allocation for each player could allow varying amounts of resources to be invested in each target, further increasing the complexity of the action space.

\section{Acknowledgements}
The authors would like to thank Rolls Royce for funding this project in part.

\bibliographystyle{IEEEarb}
\bibliography{IEEEabrv,refs}

\begin{thebibliography}{10}

\bibitem{ahmad2020maximizing}
Ishfaq Ahmad, Addison Clark, Alex Sabol, David Ferris, and Alex Aved.
\newblock Maximizing resilience under defender attacker model in heterogeneous
  multi-networks.
\newblock In {\em 2020 3rd International Conference on Data Intelligence and
  Security (ICDIS)}, pages 117--126. IEEE, 2020.

\bibitem{berner2019dota}
Christopher Berner, Greg Brockman, Brooke Chan, Vicki Cheung, Przemys{\l}aw
  D{{e}}biak, Christy Dennison, David Farhi, Quirin Fischer, Shariq Hashme,
  Chris Hesse, et~al.
\newblock Dota 2 with large scale deep reinforcement learning.
\newblock {\em arXiv preprint arXiv:1912.06680}, 2019.

\bibitem{borodin2010threshold}
Allan Borodin, Yuval Filmus, and Joel Oren.
\newblock Threshold models for competitive influence in social networks.
\newblock In {\em Internet and Network Economics: 6th International Workshop,
  WINE 2010, Stanford, CA, USA, December 13-17, 2010. Proceedings 6}, pages
  539--550. Springer, 2010.

\bibitem{kerry2022northcarolina}
Kerry Breen.
\newblock North carolina attacks underscore power grid vulnerabilities:
  `destroying this infrastructure can have a crippling effect'.
\newblock {\em CBS News}, 2022.

\bibitem{buesing2018woulda}
Lars Buesing, Theophane Weber, Yori Zwols, Sebastien Racaniere, Arthur Guez,
  Jean-Baptiste Lespiau, and Nicolas Heess.
\newblock Woulda, coulda, shoulda: Counterfactually-guided policy search.
\newblock {\em arXiv preprint arXiv:1811.06272}, 2018.

\bibitem{chaoqi2021attack}
Fu~Chaoqi, Gao Yangjun, Zhong Jilong, Sun Yun, Zhang Pengtao, and Wu~Tao.
\newblock Attack-defense game for critical infrastructure considering the
  cascade effect.
\newblock {\em Reliability Engineering \& System Safety}, 216:107958, 2021.

\bibitem{chen2022gccad}
Bo~Chen, Jing Zhang, Xiaokang Zhang, Yuxiao Dong, Jian Song, Peng Zhang, Kaibo
  Xu, Evgeny Kharlamov, and Jie Tang.
\newblock Gccad: Graph contrastive learning for anomaly detection.
\newblock {\em IEEE Transactions on Knowledge and Data Engineering}, 2022.

\bibitem{chen2022distribution}
Kaixuan Chen, Jie Song, Shunyu Liu, Na~Yu, Zunlei Feng, Gengshi Han, and Mingli
  Song.
\newblock Distribution knowledge embedding for graph pooling.
\newblock {\em IEEE Transactions on Knowledge and Data Engineering}, 2022.

\bibitem{chen2022game}
Shun Chen, Xudong Zhao, Zhilong Chen, Benwei Hou, and Yipeng Wu.
\newblock A game-theoretic method to optimize allocation of defensive resource
  to protect urban water treatment plants against physical attacks.
\newblock {\em International Journal of Critical Infrastructure Protection},
  36:100494, 2022.

\bibitem{ding2019novel}
Xiaofeng Ding, Cui Wang, Kim-Kwang~Raymond Choo, and Hai Jin.
\newblock A novel privacy preserving framework for large scale graph data
  publishing.
\newblock {\em IEEE transactions on knowledge and data engineering},
  33(2):331--343, 2019.

\bibitem{gai2010contagion}
Prasanna Gai and Sujit Kapadia.
\newblock Contagion in financial networks.
\newblock {\em Proceedings of the Royal Society A: Mathematical, Physical and
  Engineering Sciences}, 466(2120):2401--2423, 2010.

\bibitem{ghena2014green}
Branden Ghena, William Beyer, Allen Hillaker, Jonathan Pevarnek, and J~Alex
  Halderman.
\newblock Green lights forever: Analyzing the security of traffic
  infrastructure.
\newblock {\em WOOT}, 14:7--7, 2014.

\bibitem{guo2020enhancing}
Yangyang Guo, Zhiyong Cheng, Jiazheng Jing, Yanpeng Lin, Liqiang Nie, and Meng
  Wang.
\newblock Enhancing factorization machines with generalized metric learning.
\newblock {\em IEEE Transactions on Knowledge and Data Engineering},
  34(8):3740--3753, 2020.

\bibitem{guo2021reinforcement}
Youqi Guo, Lingfeng Wang, Zhaoxi Liu, and Yitong Shen.
\newblock Reinforcement-learning-based dynamic defense strategy of multistage
  game against dynamic load altering attack.
\newblock {\em International Journal of Electrical Power \& Energy Systems},
  131:107113, 2021.

\bibitem{heinrich2016deep}
Johannes Heinrich and David Silver.
\newblock Deep reinforcement learning from self-play in imperfect-information
  games.
\newblock {\em arXiv preprint arXiv:1603.01121}, 2016.

\bibitem{huang2021using}
Wencheng Huang, Bowen Zhou, Yaocheng Yu, Hao Sun, and Pengpeng Xu.
\newblock Using the disaster spreading theory to analyze the cascading failure
  of urban rail transit network.
\newblock {\em Reliability Engineering \& System Safety}, 215:107825, 2021.

\bibitem{korkali2017reducing}
Mert Korkali, Jason~G Veneman, Brian~F Tivnan, James~P Bagrow, and Paul~DH
  Hines.
\newblock Reducing cascading failure risk by increasing infrastructure network
  interdependence.
\newblock {\em Scientific reports}, 7(1):1--13, 2017.

\bibitem{li2019attack}
Yapeng Li, Ye~Deng, Yu~Xiao, and Jun Wu.
\newblock Attack and defense strategies in complex networks based on game
  theory.
\newblock {\em Journal of Systems Science and Complexity}, 32(6):1630--1640,
  2019.

\bibitem{li2018influence}
Yuchen Li, Ju~Fan, Yanhao Wang, and Kian-Lee Tan.
\newblock Influence maximization on social graphs: A survey.
\newblock {\em IEEE Transactions on Knowledge and Data Engineering},
  30(10):1852--1872, 2018.

\bibitem{lu2020sample}
Chaochao Lu, Biwei Huang, Ke~Wang, Jos{\'e}~Miguel Hern{\'a}ndez-Lobato, Kun
  Zhang, and Bernhard Sch{\"o}lkopf.
\newblock Sample-efficient reinforcement learning via counterfactual-based data
  augmentation.
\newblock {\em arXiv preprint arXiv:2012.09092}, 2020.

\bibitem{lu2018fast}
Wei Lu, Yanyan Shen, Tongtong Wang, Meihui Zhang, Hosagrahar~Visvesvaraya
  Jagadish, and Xiaoyong Du.
\newblock Fast failure recovery in vertex-centric distributed graph processing
  systems.
\newblock {\em IEEE Transactions on Knowledge and Data Engineering},
  31(4):733--746, 2018.

\bibitem{paul2019learning}
Shuva Paul, Zhen Ni, and Chaoxu Mu.
\newblock A learning-based solution for an adversarial repeated game in
  cyber--physical power systems.
\newblock {\em IEEE transactions on neural networks and learning systems},
  31(11):4512--4523, 2019.

\bibitem{pitis2020counterfactual}
Silviu Pitis, Elliot Creager, and Animesh Garg.
\newblock Counterfactual data augmentation using locally factored dynamics.
\newblock {\em Advances in Neural Information Processing Systems},
  33:3976--3990, 2020.

\bibitem{silver2016mastering}
David Silver, Aja Huang, Chris~J Maddison, Arthur Guez, Laurent Sifre, George
  Van Den~Driessche, Julian Schrittwieser, Ioannis Antonoglou, Veda
  Panneershelvam, Marc Lanctot, et~al.
\newblock Mastering the game of go with deep neural networks and tree search.
\newblock {\em nature}, 529(7587):484--489, 2016.

\bibitem{silver2017mastering}
David Silver, Thomas Hubert, Julian Schrittwieser, Ioannis Antonoglou, Matthew
  Lai, Arthur Guez, Marc Lanctot, Laurent Sifre, Dharshan Kumaran, Thore
  Graepel, et~al.
\newblock Mastering chess and shogi by self-play with a general reinforcement
  learning algorithm.
\newblock {\em arXiv preprint arXiv:1712.01815}, 2017.

\bibitem{vinyals2019grandmaster}
Oriol Vinyals, Igor Babuschkin, Wojciech~M Czarnecki, Micha{\"e}l Mathieu,
  Andrew Dudzik, Junyoung Chung, David~H Choi, Richard Powell, Timo Ewalds,
  Petko Georgiev, et~al.
\newblock Grandmaster level in starcraft ii using multi-agent reinforcement
  learning.
\newblock {\em Nature}, 575(7782):350--354, 2019.

\bibitem{wang2023attack}
Shuliang Wang, Jingya Sun, Jianhua Zhang, Qiqi Dong, Xifeng Gu, and Chen Chen.
\newblock Attack-defense game analysis of critical infrastructure network based
  on cournot model with fixed operating nodes.
\newblock {\em International Journal of Critical Infrastructure Protection},
  40:100583, 2023.

\bibitem{wei2016stochastic}
Longfei Wei, Arif~I Sarwat, Walid Saad, and Saroj Biswas.
\newblock Stochastic games for power grid protection against coordinated
  cyber-physical attacks.
\newblock {\em IEEE Transactions on Smart Grid}, 9(2):684--694, 2016.

\bibitem{xu2019network}
Linchuan Xu, Jiannong Cao, Xiaokai Wei, and S~Yu Philip.
\newblock Network embedding via coupled kernelized multi-dimensional array
  factorization.
\newblock {\em IEEE Transactions on Knowledge and Data Engineering},
  32(12):2414--2425, 2019.

\bibitem{yan2016q}
Jun Yan, Haibo He, Xiangnan Zhong, and Yufei Tang.
\newblock Q-learning-based vulnerability analysis of smart grid against
  sequential topology attacks.
\newblock {\em IEEE Transactions on Information Forensics and Security},
  12(1):200--210, 2016.

\bibitem{zeng2022transportation}
Ziqiang Zeng, Lei Tan, Zhuo Chen, and Yurui Chang.
\newblock Transportation networks cascading failures modeling under emergency
  environment.
\newblock In {\em Proceedings of the Sixteenth International Conference on
  Management Science and Engineering Management--Volume 1}, pages 43--56.
  Springer, 2022.

\bibitem{zhang2021labeling}
Muhan Zhang, Pan Li, Yinglong Xia, Kai Wang, and Long Jin.
\newblock Labeling trick: A theory of using graph neural networks for
  multi-node representation learning.
\newblock {\em Advances in Neural Information Processing Systems},
  34:9061--9073, 2021.

\bibitem{zhang2019graph}
Si~Zhang, Hanghang Tong, Jiejun Xu, and Ross Maciejewski.
\newblock Graph convolutional networks: a comprehensive review.
\newblock {\em Computational Social Networks}, 6(1):1--23, 2019.

\end{thebibliography}

\end{document}